%% file: lattice_polytopes.tex
\documentclass[10pt,a4paper]{article}
\usepackage[english]{babel}
\usepackage{amsmath}
\usepackage{enumerate}
\usepackage{amssymb}
\usepackage{amsthm}
\usepackage{comment}
\usepackage{psfrag}
\usepackage{tikz}
\usepackage{mathtools}
\usepackage[pdfstartview=FitB]{hyperref}

\usepackage{graphicx}
\usepackage{caption}
\usepackage{subcaption}
 
\usepackage{color}
\definecolor{white}{RGB}{255,255,255}
\definecolor{light}{RGB}{229,229,229}
\definecolor{gray}{RGB}{191,191,191}
\definecolor{dark}{RGB}{153,153,153}
\definecolor{truered}{rgb}{1,.0,.0}
\newcommand{\red}[1]{\textcolor{truered}{#1}}
\newcommand{\invisible}[1]{\textcolor{white}{#1}}

\usepackage[margin=1in]{geometry}

\usepackage[ruled,vlined]{algorithm2e}

\newtheorem{theorem}{Theorem}
\newtheorem{corollary}[theorem]{Corollary}
\newtheorem{lemma}{Lemma}

\DeclarePairedDelimiter{\floor}{\lfloor}{\rfloor}

\newcommand{\bb}{\mathbb}
\newcommand{\R}{\bb R}

\newcommand{\N}{\bb N}

\SetKwFunction{BP}{kPath}
\SetKwFunction{SP}{SimplexPath}
\SetKwFunction{SS}{SimplexStep}
\SetKwFunction{ZP}{1Path}
\SetKwInOut{Input}{input}
\SetKwInOut{Output}{output}

\def\01{$(0,1)$}

\renewcommand{\int}{\mathop{\mathrm{int}}}

\newcommand{\conv}{\mathop{\mathrm{conv}}}

\newcommand{\st}{:}

\newcommand{\skipp}{\smallskip}

\newcounter{claim}[section]
\newenvironment{claim}[1][]
{\refstepcounter{claim} \begin{trivlist} \item[] {\bf Claim~\theclaim}\space#1 \itshape}
{\end{trivlist}}

\newenvironment{cpf}
{\begin{trivlist} \item[] {\em Proof of claim.}}{
$\diamond$ \end{trivlist}}

\newcounter{mynotes}
\input{fig-utils}

\begin{document}

%\title{On the diameter of half-integral polytopes}
\title{On the diameter of lattice polytopes}
\author{Alberto Del Pia\thanks{Department of Industrial and Systems Engineering, University of Wisconsin-Madison. Email: \url{delpia@wisc.edu}}%
\and Carla Michini\thanks{Wisconsin Institute for Discovery, University of Wisconsin-Madison. Email: \url{michini@wisc.edu}}%
}

%\author{Alberto Del Pia \and Carla Michini}
%\thanks{
%Department of Industrial and Systems Engineering \& Wisconsin Institute for Discovery,
%University of Wisconsin-Madison, USA.
%E-mail: {\tt delpia@wisc.edu}.}

%\thanks{
%Wisconsin Institute for Discovery,
%University of Wisconsin-Madison, USA.
%E-mail: {\tt michini@wisc.edu}.}
%}

%\date{\today}

\maketitle

\begin{abstract}
%Half-integral polytopes play a crucial role in discrete optimization, since they often appear as relaxations of polytopes arising from combinatorial optimization problems.
In this paper we show that the diameter of a $d$-dimensional lattice polytope in $[0,k]^n$ is at most  $\floor*{\left(k-\frac{1}{2}\right)d}$.
This result implies that the diameter of a $d$-dimensional half-integral polytope is at most $\floor*{\frac 32 d}$.
We also show that for half-integral polytopes the latter bound is tight for any $d$.
%We also give a constructive proof to show that the diameter 
%We also give a constructive proof of a more general upper bound of $\floor*{\left(k-\frac{1}{2}\right)d}$ for the diameter of a $d$-dimensional lattice polytope $P \subseteq [0,k]^n$.
%, which improves on the known bound of $kd$ given by Kleinschmidt and Onn.
%Our proof is constructive, and it yields an algorithm to find a path between two given vertices of $P$, whose length does not exceed our bound.
%The running time of our algorithm is bounded by a polynomial in $k$, and in the size of the given external description of $P$.
\end{abstract}

\section{Introduction}

The \emph{$1$-skeleton} of a polyhedron $P$ is the graph whose nodes are the vertices of $P$, and that has an edge joining two nodes if and only if the corresponding vertices of $P$ are adjacent on $P$.
Given vertices $u,v$ of $P$, the \emph{distance} $\delta^P(u,v)$ between $u$ and $v$ is the length of a shortest path connecting $u$ and $v$ on the $1$-skeleton of $P$.
We may write $\delta(u,v)$ instead of $\delta^P(u,v)$ when the polyhedron we are referring to is clear from the context.
The \emph{diameter} $\delta(P)$ of $P$ is the smallest  number that bounds the distance between any pair of vertices of $P$.
%We refer to~\cite{Sch03} for basic notions and terminology on polyhedra and graphs.

%One fundamental question in linear optimization and discrete geometry concerns whether it is possible to bound the diameter of a $d$-dimensional polyhedron defined by $m$ facets with a polynomial in $d$ and $m$. This question is related to the Hirsch conjecture (1957), which stated that a $d$-dimensional polyhedron with $m$ facets cannot have diameter greater than $m-d$.
%The Hirsch conjecture was first disproven by Klee and Walkup~\cite{KleWal67} for unbounded polyhedra, and only recently by Santos~\cite{San11} for polytopes.
%It is still an open question whether the diameter can be bounded by a polynomial function of $d$ and $m$.
%The motivation for bounding the diameter of a polyhedron lies in its connection with the performance of the simplex method, as the diameter is a lower bound on the number of pivots that the simplex algorithm may need to reach an optimal vertex.
%Sometimes, the special structure of a polytope can be exploited to bound its diameter (see, \eg, \cite{balinskiRussakoff74,balinskirispoli93,balinski84,brightwell06,MicSas14}).

In this paper, we investigate the diameter of \emph{lattice polytopes}, i.e.~polytopes whose vertices are integral.
%Our motivation is that 
Lattice polytopes play a crucial role in discrete optimization and integer programming problems, where the variables are constrained to assume integer values.
%Moreover, given a rational polytope $P$ such that $\lambda$ is the least common denominator of the coordinates of the vertices of $P$, we have that $\lambda P$ is a lattice polytope.
%
Our goal is to define a bound on the diameter of a lattice polytope $P$, that depends on the dimension of $P$ and on the parameter $k = \max\{||x-y||_{\infty} \st x,y \in P\}$, in order to apply such bound to classes of polytopes for which $k$ is known to be small.
A similar approach has been followed by Bonifas et al. \cite{BonDiSEisHanNie14}, who showed that the diameter of a polyhedron $P = \{x \in \R^n \st Ax \le b\}$ is bounded by a polynomial that depends on $n$ and on the parameter $\Delta$, defined as the largest absolute value of any sub-determinant of $A$. Note that, while $\Delta$ is related to the external description of $P$, $k$ is related to its internal description. However, both $\Delta$ and $k$ are in general not polynomial in $n$ and in the number of the facet-defining inequalities of $P$.

For $k \in \N$, a \emph{$(0,k)$-polytope} $P \subseteq \R^n$ is a lattice polytope contained in $[0,k]^n$.
Naddef~\cite{Nad89} showed that the diameter of a $d$-dimensional \01-polytope is at most $d$, and this bound is tight for the hypercube $[0,1]^d$.
Kleinschmidt and Onn~\cite{KleOnn92} extended this result by proving that the diameter of a $d$-dimensional $(0,k)$-polytope cannot exceed $kd$. However, their bound is not tight for $k \ge 2$.

Our main contribution is establishing an upper bound for the diameter of a $d$-dimensional $(0,k)$-polytope, which refines the 
% 20-years old 
bound by Kleinschmidt and Onn.
\begin{theorem}
\label{th: main}
For $k \ge 2$, the diameter of a $d$-dimensional $(0,k)$-polytope is at most $\floor*{\left(k-\frac{1}{2}\right)d}$.
\end{theorem}
%
%\begin{proof}[Proof of Corollary~\ref{k=2}]
%This section is devoted to proving that the bound given in Theorem~\ref{th: main} is tight for $k=2$.
%
%Theorem~\ref{th: main} (with $k=2$), immediately implies that the diameter of a $d$-dimensional $(0,2)$-polytope is at most $\floor*{\frac 32 d}$.
%For any natural number $d$, we construct a $d$-dimensional $(0,2)$-polytope whose diameter equals $\floor*{\frac 32 d}$.
%\end{proof}
%
The proof of Theorem~\ref{th: main} is elementary, as it combines an induction argument with basic tools from linear programming and polyhedral theory. 
%Our proof is ,
%technique provides new insight on the structure of lattice polytopes,
%Moreover it allows us to . % for the diameter of lattice polytopes.
%Moreover, our proof is constructive, and it yields an algorithm that finds a path between two given vertices, whose length does not exceed $\floor*{\left(k-\frac{1}{2}\right)d}$.
%For a $(0,k)$-polytope $P = \{x : Ax \le b\}$, the algorithm that we propose runs in time bounded by a polynomial in $k$, and in the size of $A,b$.
Our proof is also constructive, since it shows how to build a path between two given vertices of $P$, whose length does not exceed our bound.
%Our proof is also constructive, and it yields an algorithm to find a path between two given vertices of $P$, whose length does not exceed our bound.
%The running time of the algorithm is bounded by a polynomial in $k$, and in the size of the given external description of $P$.

For $(0,2)$-polytopes, we show that the upper bound given in Theorem~\ref{th: main} is tight for any $d$.
%Our main contribution is establishing a bound for the diameter of a $d$-dimensional $(0,2)$-polytope that is tight for any $d$. 
\begin{corollary}
\label{k=2}
The diameter of a $d$-dimensional $(0,2)$-polytope is at most $\floor*{\frac 32 d}$.
Moreover, for any natural number $d$, there exists a $d$-dimensional $(0,2)$-polytope attaining this bound.
\end{corollary}
%While the upper bound in Corollary~\ref{k=2} follows from Theorem~\ref{th: main}, the lower bound 
%Clearly, $P$ is a half-integral polytope if and only if $2P = \{2x \st x \in P\}$ is a $(0,2)$-polytope.
%As $P$ and $2P$ have the same $1$-skeleton and the same dimension,
%Corollary~\ref{k=2} implies that the diameter of a half-integral polytope is at most $\floor*{\frac 32 d}$, and the bound is tight.
%The strength of Corollary~\ref{k=2} is that there is an easy construction to define in any dimension a half-integral polytope attaining our bound. This construction is based on the cartesian products of polytopes of dimension $1$ and $2$.
The lower bound of Corollary~\ref{k=2} 
follows by an easy construction based on the cartesian product of polytopes  of dimension one and two.
It is well-known that, given two polytopes $P_1$ and $P_2$, their cartesian product $P_1 \times P_2$ satisfies $\delta(P_1 \times P_2) = \delta(P_1) +  \delta(P_2)$.
Now, let $H_{1} = [0,2] $ and $H_{2} = \conv \{(0,0), (1,0), (0,1), (2,1), (1,2), (2,2)\}.$
For even $d$, let $H_d = (H_2)^{d/2}$, and for odd $d$, let $H_d = H_{d-1} \times H_1$.
Thus for all $d \in \N$, $H_{d}$ is a $d$-dimensional $(0,2)$-polytope, with $\delta(H_{d}) = \floor*{\frac{3}{2}d}$.

Corollary~\ref{k=2} has important implications for the diameter of half-integral polytopes.
\emph{Half-integral polytopes} are polytopes whose vertices have components in $\left\{0,\frac{1}{2},1\right\}$, and they are affinely equivalent to $(0,2)$-polytopes.
The class of half-integral polytopes is very rich, as many half-integral polytopes appear in the literature as relaxations of \01-polytopes arising from combinatorial optimization problems. In some cases, while the \01-polytope defined as the convex hull of the feasible solutions to the combinatorial problem has exponentially many facets, there is a linear relaxation, defined by a polynomial number of constraints, that yields a half-integral polytope.
%Gaining insights on the structure of half-integral polytopes is also relevant for approximation algorithms.
%In fact, if $P = \{x \in \R^n : Ax \ge b, \; x \ge 0 \}$ is a half-integral polytope, then for any non-negative $A$ and $c \in \R^n$, one can approximate the optimum of $\min\{cx : x \in P \cap \{0,1\}^n \}$ within a relative factor of $2$, by solving the linear programming relaxation $\min\{cx : x \in P \}$, and rounding to $1$ all variables with value at least $1/2$.

There are several classes of polytopes that are known to be half-integral, such as the fractional matching polytope and the fractional stable set polytope \cite{Bal65}, the linear relaxation of the boolean quadric polytope and the rooted semimetric polytope \cite{Pad89} (see also \cite{Sch03} and \cite{DezLau97}).
An interesting class of half-integral polytopes arises from totally dual half-integral systems,
%A linear system $Ax \le b$ is called \emph{totally dual half-integral (TDI/2)} if the minimum of the dual problem of $\max\{cx : Ax \le b\}$ has a half-integral optimal solution \footnote{A vector $y$ is called \emph{half-integral} if all coordinates of $2y$ are integral.}, for every integral vector $c$ for which the optimum is finite (see e.g.~\cite{Sch03}).
%From the Edmonds-Giles theorem \cite{EdmGil77}, we deduce that if $Ax \le b$ is TDI/2 and $b$ is integral, then $\max\{cx : Ax \le b\}$  has a half-integral optimal solution, for every integral vector $c$ for which the optimum is finite.
%Thus, a polytope contained in $[0,1]^n$ described by a TDI/2 system is half-integral.
such as the fractional stable matching polytope \cite{AbeRot94,CheDinHuZan12}, and the fractional matroid matching polytope \cite{Sch86,GijPap13}.

The rest of the paper is devoted to proving Theorem~\ref{th: main}.
%In Section~\ref{sec: ingredients} we present some lemmas that will be crucial to prove our main results.
%, and we give a short proof of the bound by Kleinschmidt and Onn.
%In Section~\ref{sec: proof} we prove our upper bound for the diameter of $(0,k)$-polytopes.
% and we present an algorithm that constructs a short path between any two vertices.
%In Section~\ref{sec: tight} we show that our bound is tight for $k=2$.
%Finally, in Section~\ref{sec: conclusions}, some conclusions are drawn.

\section{Proof of main result}
In order to bound the diameter of a non full-dimensional $(0,k)$-polytope $P \subseteq \R^n$,
we define the \emph{projection of $P$ onto the $i$-coordinate hyperplane} as the polytope 
$$\{ \bar x \in \R^{n-1}: \exists \; x \in P \text{ with } x_j=\bar x_j \text{ for } j=1,\dots,i-1, \, x_j = \bar x_{j-1}  \text{ for } j=i+1,\dots,n\}.$$
That is, we simply drop the $i$-th coordinate from all vectors in $P$.
Since integral vectors are mapped into integral vectors, the next lemma follows from Theorem 3.3 in \cite{NadPul84}.
\begin{lemma}
\label{lem: full dim}
Let $P$ be a $d$-dimensional $(0,k)$-polytope in $\R^n$ with $d \ge 1$.
Then there exists a full-dimensional $(0,k)$-polytope in $\R^d$ with the same $1$-skeleton as $P$.
\end{lemma}

%Before proving Theorem~\ref{th: main}, we present some lemmas
For $d,k\in \N$, we define $\delta_k^d$ to be the maximum possible diameter of a $(0,k)$-polytope of dimension at most $d$, i.e.
$$
\delta_k^d = \max \{\delta(P) \st \text{$P$ is a $(0,k)$-polytope of dimension at most $d$}\}.
$$
Note that the maximum in the definition of $\delta_k^d$ always exists.
In fact, it follows from Lemma~\ref{lem: full dim} that the number of vertices of a $d$-dimensional $(0,k)$-polytope is at most $(k+1)^d$, thus also its diameter is upper bounded by $(k+1)^d$, which is a number independent on the dimension of the ambient space of $P$.
Moreover, for fixed $k$, the value $\delta_k^d$ is clearly non-decreasing in $d$.

We now present some lemmas that will be used to prove Theorem~\ref{th: main}.
These results follow by applying the ideas introduced by Kleinschmidt and Onn in \cite{KleOnn92}.
%We first present some lemmas that will be crucial to prove our main result.
The next lemma shows how to bound the distance $\delta(u,F)$ between a vertex $u$ of a lattice polytope $P$ and a face $F$ of $P$, that is defined as $\delta(u,F)=\min \{\delta(u,v) : v \text{ is a vertex of } F\}$.
We say that two vertices $u,v$ of a polytope are \emph{neighbors} if $\delta(u,v)=1$.
We denote by $e^i$, for $i=1,\dots,n$, the $i$-th vector of the standard basis of $\R^n$.

\begin{lemma}
\label{lem: bounds}
Let $P$ be a lattice polytope, and let $u$ be a vertex of $P$.
Let $c$ be an integral vector, $\gamma=\min \{cx \st x \in P\}$, and $F =\{x \in P \st cx=\gamma\}$.
%Let $F =\{x \in P \st cx=\gamma\}$ be a face of $P$ such that $\gamma=\min \{cx \st x \in P\}$.
%Then there exists a vertex $u'$ of $F$ such that $\delta(u,u') \le cu-\gamma$.
Then $\delta(u,F) \le cu-\gamma$.
\end{lemma}

\begin{proof}
We show that there exists a vertex $v$ of $F$ such that $\delta(u,v) \le cu-\gamma$.
We prove this statement by induction on the integer value $cu-\gamma \ge 0$.
The statement is trivial for $cu-\gamma =0$, as we can set $v=u$.
Assume $cu-\gamma \ge 1$.
%OLD:
%Because $F$ is nonempty, there exists a vector $v \in F$ such that $cu-cv=cu-\gamma \ge 1$.
%Thus, by Lemma~\ref{lem: levels}, there exists a neighbor $u'$ of $u$ with $cu' \le cu-1$.
%NEW:
Since $F$ is nonempty, there exists a neighbor $u'$ of $u$ with $cu' < cu$ 
(see, e.g., \cite{Bro83}).
The integrality of $c$, $u'$ and $u$, implies $cu' \le cu-1$.
As $cu'-\gamma \le cu-\gamma -1$,
by the induction hypothesis there exists a vertex $v$ of $F$ such that $\delta(u',v) \le cu'-\gamma$.
Therefore $\delta(u,v) \le \delta(u,u') + \delta(u',v) \le 1+ cu'-\gamma \le cu-\gamma$.
\end{proof}
 
Given two vertices $u$ and $v$ and a face $F$ of a lattice polytope $P$, we have $\delta(u,v) \le \delta(u,F) + \delta(v,F) + \delta(F)$.
%If $F =\{x \in P \st cx=\gamma\}$ is such that $\gamma=\min \{cx \st x \in P\}$,
By applying Lemma~\ref{lem: bounds} to both $u$ and $v$, we obtain an upper bound on $\delta(u,v)$ that depends on $F$: 
%$\delta(u,v) \le \delta(F) + cu + cv - 2\gamma$.

\begin{lemma} \label{lem: ell}
Let $P$ be a lattice polytope, and let $u,v$ be vertices of $P$.
Let $c$ be an integral vector, $\gamma=\min \{cx \st x \in P\}$, and $F =\{x \in P \st cx=\gamma\}$.
Then $\delta(u,v) \le \delta(F) + cu + cv - 2\gamma$.
\end{lemma}

%If $P$ is a $(0,k)$-polytope in $\R^n$ and for $i \in \{1,\dots, n\}$ we define $l=\min\{x_i :x \in P\}$ and $h=\max\{x_i :x \in P\}$,

Let $P$ be a $(0,k)$-polytope in $\R^n$ and let $l=\min\{x_i :x \in P\}$ and $h=\max\{x_i :x \in P\}$ for some $i \in \{1,\dots, n\}$. 
%$x_i \in [l,h]$ for all $x \in P$,
We can bound the distance between any two vertices $u$ and $v$ of $P$ by bounding their distances from the faces $L=\{x \in P : x_i = l\}$ and $H=\{x \in P : x_i = h\}$. If $u_i + v_i \le l+h$, Lemma \ref{lem: ell} applied with $F=L$, $c = e^i$ and $\gamma = l$ implies $\delta(u,v) \le \delta(L) + (h-l)$. If $u_i + v_i \ge l+h$, Lemma \ref{lem: ell} applied with $F=H$, $c = -e^i$ and $\gamma = -h$ implies $\delta(u,v) \le \delta(H) + (h-l)$. Since $L$ and $H$ are $(0,k)$-polytopes of dimension at most $n-1$, we have that both $\delta(L)$ and $\delta(H)$ are at most $\delta_k^{n-1}$.

\begin{lemma}
\label{lem: ind}
Let $P$ be a $(0,k)$-polytope in $\R^n$, and suppose that there exists $i \in \{1, \dots, n\}$ such that  $x_i \in [l,h]$ for every $x \in P$.
Then $\delta(P) \le  \delta_k^{n-1} + (h-l)$.
\end{lemma}

%A $d$-dimensional $(0,k)$-polytope $P \subseteq \R^n$ is such that $x_i \in [0,k]$ for $i=1,\dots,n$.

%proved the bound $\delta(P) \le kd$ by applying Lemma~\ref{lem: bounds} with $c = \pm e^i$\note{In the following, I omit $i \in \{1,\dots,d\}$} and by using induction on $d$. The key step of their proof consists in showing that, for any two vertices $u$ and $v$ of $P$, $\delta(u,v) \le \delta(F) + k$, where $F = \{x \in P : x_i = \beta\}$ is a proper face of $P$\note{proper?}. In fact, if $u_i + v_i \le k$, by applying Lemma~\ref{lem: bounds} with $c = e^i$ and $\gamma = \beta \ge 0$, it follows $\delta(u,v) \le \delta(F) + u_i + v_i - 2\gamma \le \delta(F) + k$. If $u_i + v_i \ge k$, Lemma~\ref{lem: bounds} with $c=-e^i$ and $\gamma = - \beta \ge -k$ implies $\delta(u,v) \le  \delta(F) - (u_i + v_i) - 2 \gamma \le \delta(F) + k$. As $F$ is a $(0,k)$-polytope of dimension at most $d-1$, $\delta(u,v) \le \delta(F) + k \le \delta_k^{d-1} + k$.
\skipp

Given a $d$-dimensional $(0,k)$-polytope $P$, Kleinschmidt and Onn prove the bound $\delta(P) \le k d$ by essentially applying Lemma \ref{lem: full dim}, and then Lemma \ref{lem: ind} inductively.
Therefore, their bound uses Lemma \ref{lem: bounds} only with vectors $c=\pm e^i$.
To prove our refined bound, we will use Lemma \ref{lem: bounds} also with different vectors $c$.
%Note that Lemma~\ref{lem: ind} and Lemma~\ref{lem: full dim} are sufficient \note{change} to show the bound by Kleinschmidt and Onn~\cite{KleOnn92} also for non full-dimensional $(0,k)$-polytopes.
We are now ready to give the proof of Theorem~\ref{th: main}.

\skipp

\begin{proof}[Proof of Theorem~\ref{th: main}]
Let $P$ be a $d$-dimensional $(0,k)$-polytope, with $k \ge 2$.
The proof is by induction on $d$.
The base cases are $d=0$ and $d=1$.
The diameter of a $0$-dimensional polytope is clearly zero, and the diameter of a $1$-dimensional polytope is at most one, thus also bounded by $\floor*{k-\frac{1}{2}} = k-1$ since $k \ge 2$.

We now assume $d \ge 2$.
Let $u,v$ be vertices of $P$.
By the induction hypothesis we assume that Theorem~\ref{th: main} is true for $(0,k)$-polytopes of dimension at most $d-1$.
In particular, $\delta_k^{d-1} \le \floor*{\left(k-\frac{1}{2}\right)(d-1)}$, and $\delta_k^{d-2} \le \floor*{\left(k-\frac{1}{2}\right) (d-2)}$.
Thus, in order to prove the inductive step, it is sufficient to show one of the following two inequalities:
\begin{align}
\delta(u,v) & \le \delta_k^{d-1}+k-1,\label{e: d-1}\\
\delta(u,v) & \le \delta_k^{d-2}+2k-1.\label{e: d-2}
\end{align}

\begin{claim}
We can assume that $P$ is full-dimensional.
\end{claim}

\begin{cpf}
By Lemma~\ref{lem: full dim}, there exists a full-dimensional $(0,k)$-polytope in $\R^d$ with the same $1$-skeleton as $P$.
%Hence, from now on, we can assume without loss of generality that $P$ is full-dimensional.
\end{cpf}

\begin{claim}
We can assume that $P$ intersects all facets of the hypercube $[0,k]^d$.
\end{claim}

\begin{cpf}
If there exists a facet $G$ of the hypercube $[0,k]^d$ with $P \cap G = \emptyset$, then let $i \in \{1,\dots,d\}$ be such that $l \le x_i \le h$, with $l \ge 1$ or $h \le k-1$.
By Lemma~\ref{lem: ind}, $\delta(u,v) \le \delta_k^{d-1} + k-1$, i.e.~(\ref{e: d-1}) is satisfied.%
%Therefore we now assume that $P \cap G$ is nonempty for every facet $G$ of the hypercube $[0,k]^d$.
\end{cpf}

\noindent
In the remainder of the paper, we will denote by $k^d$ the $d$-dimensional vector with all entries equal to $k$.

\begin{claim}
We can assume that $u+v = k^d$.
\end{claim}

\begin{cpf}
If $u+v \neq k^d$, there exists an index $i \in \{1,\dots,d\}$ such that $u_i+v_i \le k-1$ or $u_i+v_i \ge k+1$. 
By Lemma~\ref{lem: ell} applied with $c = e^i$ or $c=-e^i$, respectively, we obtain $\delta(u,v) \le \delta(F) + k-1$, where $F$ is the face of $P$ that minimizes $cx$.
As $F$ is a $(0,k)$-polytope of dimension at most $d-1$, we have $\delta(F) \le \delta_k^{d-1}$, therefore
$\delta(u,v) \le \delta_k^{d-1} + k-1$, i.e.~(\ref{e: d-1}) is satisfied.
%Therefore we now assume that $u+v = k^d$.
\end{cpf}

\begin{claim} \label{c: no intermediate}
We can assume that $u \in \{0,k\}^d$.
\end{claim}

\begin{cpf}
Assume that $u$ has one component $u_i$, $i \in \{1,\dots,d\}$, with $1 \le u_i \le k-1$.
In this case we show that \eqref{e: d-2} is satisfied.
Since the set $\{x \in P \st x_i=0\}$ is nonempty, there exists a neighbor $s$ of $u$ with $s_i < u_i$ 
%This can be seen, for example, by applying the simplex algorithm to minimize $x_i$ on $P$ starting from vertex $u$.
(see, e.g., \cite{Bro83}).
By the integrality of $s$ and $u$, this implies $s_i \le u_i-1$.
Symmetrically, since the set $\{x \in P \st x_i=k\}$ is nonempty, $u$ has a neighbor $t$ with $t_i \ge u_i+1$.
If $s_j = t_j = u_j$ for all $j \in \{1,\dots,d\}$, $j \neq i$, then by setting $\lambda = \frac{t_i - u_i}{t_i - s_i}$ we have $\lambda s + (1-\lambda) t = u$, contradicting the fact that $u$ is a vertex of $P$.
Thus, there exists an index $j \in \{1,\dots,d\}$ with $j \neq i$ such that either $s_j \neq u_j$ or $t_j \neq u_j$.
Therefore there exists a neighbor $w$ of $u$ such that $w_i \neq u_i$ and $w_j \neq u_j$, for distinct indices $i,j \in \{1,\dots,d\}$
(see Fig.~\ref{f: intermediate vertices}$(i)$).

We assume without loss of generality that $w_i < u_i$ (if not, we can perform the change of variable $\tilde x_i = k - x_i$).
Analogously, we assume $w_j < u_j$.
As $u + v = k^d$, we have $w_i+w_j+v_i+v_j \le 2k-2$.
Let $\gamma=\min \{x_i+x_j \st x \in P\}$ and $F =\{x \in P \st x_i+x_j=\gamma\}$.
By Lemma~\ref{lem: ell} (with $c = e^i+e^j$), $\delta(w,v) \le \delta(F) + w_i+w_j+v_i+v_j - 2\gamma\le \delta(F) + 2k-2 - 2\gamma$
(see Fig.~\ref{f: intermediate vertices}$(ii)$).

We now show that $\delta(F) \le \delta_k^{d-2} + \gamma$.
Let $\bar F$ be the projection of $F$ onto the $j$-coordinate hyperplane.
$\bar F$ is a $(0,k)$-polytope in $\R^{d-1}$ and, 
by
Lemma~\ref{lem: full dim}, $\bar F$ has the same $1$-skeleton of $F$.
Note that, for any $x \in F$, $x_i=\gamma-x_j$ and $x_j \ge 0$ imply $x_i \le \gamma$.
Therefore, $x_i \le \gamma$ for any $x \in \bar F$.
Then, by Lemma~\ref{lem: ind}, $\delta(\bar F) \le \delta_k^{d-2} + \gamma$, thus $\delta(F) \le \delta_k^{d-2} + \gamma$.

This implies $\delta(w,v) \le  \delta_k^{d-2} + 2k -2 - \gamma$ and,
since  $\gamma \ge 0$ and $\delta(u,w) = 1$, finally $\delta(u,v) \le \delta(u,w) + \delta(w,v) \le \delta_k^{d-2} + 2k -1$, i.e.~(\ref{e: d-2}) is satisfied.
%Thus we now assume that %$u \in \{0,k\}^d$.
%for every $i \in \{1,\dots,d\}$ we have $u_i \in \{0,k\}$.
\end{cpf}

By possibly performing the change of variable $\tilde x_1 = k - x_1$, we can further assume without loss of generality that $u_1=k$, and $v_1=0$.

Let $F$ be the face of $P$ defined by $F = \{ x \in P \st x_1 = 0 \}$.
$F$ is a $(0,k)$-polytope of dimension at most $d-1$, thus $\delta(F) \le \delta_k^{d-1}$.
By Lemma~\ref{lem: bounds} (with $c=e^1$), there exists a vertex $u'$ of $F$ such that $\delta(u,u') \le k$.
Observe that both $u'$ and $v$ lie in $F$ and therefore $\delta(u',v) \le \delta_k^{d-1}$.

If $u' = (0,u_2,\dots,u_d)$, then $u$ and $u'$ are adjacent vertices of the hypercube $[0,k]^d$, implying that $\conv\{u,u'\}$ is an edge of $[0,k]^d$ 
(see Fig.~\ref{f: vertices hc}$(i)$).
As $P$ is convex and it is contained in $[0,k]^d$, it follows that $\conv\{u,u'\}$ is also an edge of $P$. Therefore, $\delta(u,u') = 1$ and consequently $\delta(u,v) \le \delta_k^{d-1} + 1$. As $k \ge 2$, it follows $\delta(u,v) \le \delta_k^{d-1} + k -1$, i.e.~(\ref{e: d-1}) is satisfied.

Thus we now assume $u' \neq (0,u_2,\dots,u_d)$ (see Fig.~\ref{f: vertices hc}$(ii)$).
Then, there exists an index $i \in \{2,\dots,d\}$ such that $u'_i+v_i \le k-1$ or $u'_i+v_i \ge k+1$.
We assume without loss of generality that $u'_i+v_i \le k-1$ (if not, we can perform the change of variable $\tilde x_i = k - x_i$).
Let $\gamma=\min \{x_i \st x \in F\}$, $F' =\{x \in F \st x_i=\gamma\}$.
$F'$ is a $(0,k)$-polytope, and it has dimension at most $d-2$ because it is contained in the intersection of the two linearly independent hyperplanes $\{x \in \R^d \st x_1=0\}$ and $\{x \in \R^d \st x_i=\gamma\}$.
It follows that $\delta(F') \le \delta_k^{d-2}$.
Then, by applying Lemma~\ref{lem: ell} to the polytope $F$ and the vertices $u'$ and $v$, we have $\delta(u',v) \le \delta(F') + u'_i+v_i \le \delta_k^{d-2} + k-1$.
This implies $\delta(u,v) \le \delta(u,u') + \delta(u',v) \le \delta_k^{d-2}+2k-1$, i.e.~(\ref{e: d-2}) is satisfied.
\end{proof}
%\qed

%The proof of Theorem~\ref{th: main} is constructive, and it yields an algorithm that finds a path between two given vertices, whose length does not exceed $\floor*{\left(k-\frac{1}{2}\right)d}$.
%For a $(0,k)$-polytope $P = \{x : Ax \le b\}$, such algorithm runs in time bounded by a polynomial in $k$, and in the size of $A,b$.

\section{Further directions}%
%\label{sec: conclusions}
%The main contribution of this work is an upper bound of $\floor*{\frac{3}{2}d}$ for the diameter of $d$-dimensional half-integral polytopes. This bound is tight, as for any $d$ we can explicitly construct a $d$-dimensional polytope of diameter $\floor*{\frac{3}{2}d}$.
%Our upper bound for half-integral polytopes derives from a more general bound of $\floor*{\left(k-\frac{1}{2}\right)d}$ on the diameter of $(0,k)$-polytopes. The latter bound refines the known bound of $kd$ given by Kleinschmidt and Onn~\cite{KleOnn92}, and yields an algorithm to build a short path between two given vertices of a $(0,k)$-polytope $P = \{x : Ax \le b\}$, whose running time is polynomial in $k$, and in the size of $A,b$.
%
%
Both our upper bound and the one by Kleinschmidt and Onn are not tight for $k \ge 3$. As an example, $\delta^2_3 =4$, as the maximum diameter of a lattice polygon in $[0,3]^2$ is realized by the octagon. It seems that our approach cannot be easily refined to obtain a tight upper bound for general $k$.

An interesting direction of research is to study the asymptotic behavior of the function $\delta_k^d$.
It is known that the maximum number of vertices of a 2-dimensional $(0,k)$-polytope is in $\Theta (k^{2/3})$ \cite{BalBar91}, which implies the asymptotically tight bound $\delta_k^{2} \in \Theta (k^{2/3})$.
%By Lemma~\ref{lem: ind}, $\delta_k^3 \le k + O (k^{2/3})$.
%
%By applying Lemma~\ref{lemma: diameter cartesian}, we can derive a lower bound for $\delta^d_k$, i.e.~the maximum diameter of a $d$-dimensional $(0,k)$-polytope, by taking the cartesian product of $\floor*{\frac{d}{2}}$ copies of a lattice polygon of diameter $\delta^2_k$. This yields $\delta^d_k \ge \floor*{\frac{d}{2}} \delta^2_k$.
%
%As mentioned in Remark \ref{remark}, the results in \cite{BalBar91} concerning the maximum number of vertices of a lattice polygon in $[0,k]^2$ imply that $\delta_k^{2} \in O (k^{2/3})$.
%Interestingly, in \cite{BalBar91} it is also shown that $\delta_k^{2} \in \Omega (k^{2/3})$,
%i.e.~for large $k$ there exist a constant $c$ and a lattice polygon $P \subseteq [0,k]^2$ such that $\delta(P) \ge c \cdot k^{2/3}$.
%
Using cartesian products of polytopes, it follows that $\delta_k^{d} \in \Omega (k^{2/3}d)$.
This provides an asymptotic lower bound on $\delta_k^{d}$ that is a fractional power with respect to $k$ and linear in $d$.
However, the best upper bound on $\delta_k^{d}$ is linear both in $k$ and in $d$.
In other words, there is still a significant gap between the lower and the upper bound.
%
%An open question concerns how to close, or at least to reduce, this gap.

\bibliographystyle{plain.bst}
\bibliography{biblio}
\begin{figure}[h]
\begin{center}
  \begin{tabular}{cc}
    \resizebox{6cm}{!}{\input{intermediate1}}  & \resizebox{6cm}{!}{\input{intermediate2}} \\ 
    $(i)$ & $(ii)$\\
  \end{tabular}
\end{center}
\caption{
%A pictorial representation of Claim \ref{c: no intermediate}
In Claim \ref{c: no intermediate}, 
$(i)$ we construct a neighbor $w$ of $u$ with $w_i < u_i$, and $w_j < u_j$,
$(ii)$ we use Lemma~\ref{lem: ell} with $c = e^i+e^j$ to show that $\delta(w,v) \le  \delta_k^{d-2} + 2k -2$.
}
\label{f: intermediate vertices}
\end{figure}
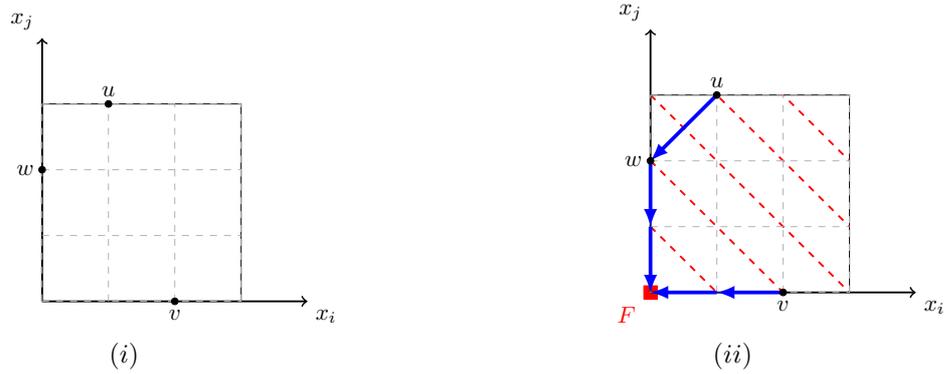
\begin{figure}[h]
\begin{center}
  \begin{tabular}{cc}
    \resizebox{6cm}{!}{\input{verticeshc2}}  & \resizebox{6cm}{!}{\input{verticeshc1}} \\ 
    $(i)$ & $(ii)$\\
  \end{tabular}
\end{center}
\caption{
To bound the distance between vertices $u \in \{0,k\}^d$ with $u_1 = k$ and $v = k^d -u$,
we construct a path from $u$ to a vertex $u'$ with $u'_1 = 0$. There are two cases:
$(i)$ $u' = (0,u_2,\dots,u_d)$, thus $\delta(u,u') =1$ and $\delta(u',v) \le \delta_k^{d-1}$;
$(ii)$ $u' \neq (0,u_2,\dots,u_d)$, thus $\delta(u,u') \le k$ and $\delta(u',v) \le \delta_k^{d-2} + k -1$.
}
\label{f: vertices hc}
\end{figure}
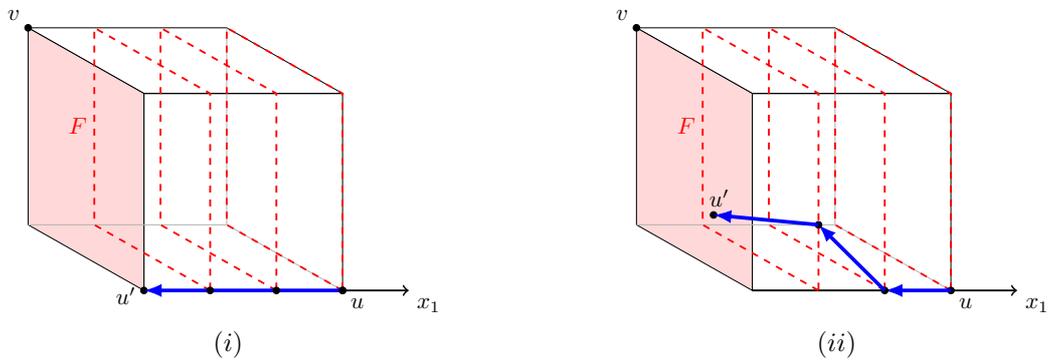
\end{document}

%% file: fig-utils.tex
\newcommand{\bdot}[1]{
\draw[color=black, fill=black] (#1) circle (.05);
}

\newcommand\Square[1]{
\draw [fill=red,red] (#1) +(-0.1,-0.1) rectangle +(0.1,0.1);
}
\newcommand\lattice[1]{
\draw[ultra thin,black!30,dashed] (0,0) grid (#1,#1);
}
\newcommand\levels[4]{

	\ifthenelse{#1 > 0 \AND #2 > 0}
	{
	\pgfmathparse{#1*#3 + #2*#3} \pgfmathresult \let\maxA\pgfmathresult 
	\foreach \a in {0,1,...,\maxA}{
	\begin{scope}
         \clip (0, #3) rectangle (#3, 0);
	\draw[color=#4,domain=0:#3,dashed,thick] plot(\x, \a/#2 - #1/#2 * \x)node[right] {};
	\end{scope}}
	\foreach \i in {0,1,...,#3}{
	\pgfmathparse{#2*\i} \pgfmathresult \let\value\pgfmathresult 
	\node [left,#4] at (0,\i) {$\pgfmathprintnumber{\value}$};}
	\foreach \j in {1,...,#3}{
	\pgfmathparse{#2*#3 + #1*\j} \pgfmathresult \let\val\pgfmathresult 
	\node [above,#4] at (\j, #3) {$\pgfmathprintnumber{\val}$};}
	}
	{
	\ifthenelse{#1>0 \OR #1<0}
	{
	\pgfmathparse{#3-1} \pgfmathresult \let\limit\pgfmathresult 
	\foreach \i in {0,...,\limit}{
	\pgfmathparse{#1*\i+#2*#3} \pgfmathresult \let\value\pgfmathresult 
	\node [above,#4] at (\i,#3) {$\pgfmathprintnumber{\value}$};}
	}
	{}
	\ifthenelse{#2>0 \OR #2<0}
	{
	\foreach \j in {0,...,#3}{
	\pgfmathparse{\j*#2 + #1*#3} \pgfmathresult \let\val\pgfmathresult 
	\node [right,#4] at (#3,\j) {$\pgfmathprintnumber{\val}$};}
	}
	{}
	\ifthenelse{#1 < 0 \AND #2 > 0}
	{
		\pgfmathparse{#1*#3} \pgfmathresult \let\maxA\pgfmathresult
		\pgfmathparse{#2*#3} \pgfmathresult \let\minA\pgfmathresult
		\foreach \a in {\minA,...,\maxA}{
		\begin{scope}
        		\clip (0, #3) rectangle (#3, 0);
		\draw[color=#4,domain=0:#3,dashed,thick] plot(\x, \a/#2 - #1/#2 * \x)node[right] {};
		\end{scope}}
	}
	{
		\pgfmathparse{#2*#3} \pgfmathresult \let\maxA\pgfmathresult
		\pgfmathparse{#1*#3} \pgfmathresult \let\minA\pgfmathresult
		\foreach \a in {\minA,...,\maxA}{
		\pgfmathparse{\a/#2} \pgfmathresult \let\b\pgfmathresult
		\pgfmathparse{#1/#2} \pgfmathresult \let\c\pgfmathresult
		\begin{scope}
         	\clip (0, #3) rectangle (#3, 0);
		\draw[color=#4,domain=0:#3,dashed,thick] plot(\x, \b - \c * \x)node[right] {};
		\end{scope}}
	}
	}
}

%% file: intermediate1.tex
\begin{tikzpicture}
\coordinate (o) at (0,0);
\coordinate (o1) at (4,0);
\coordinate (o2) at (0,4);

\coordinate (fake) at (-1.75,4);
\node (f) [above left] at (fake) {$\invisible{f}$};

\coordinate (v) at (2,0);
\coordinate (u) at (1,3);
\coordinate (u1) at (0,2);

\draw[] (o) -- (3,0) -- (3,3) -- (0,3) -- cycle;

\begin{scope}
\clip (0, 3) rectangle (3, 0);
\end{scope}

\draw[->,thick] (o) -- (o1);
\draw[->,thick] (o) -- (o2);
\node [below right] at (o1) {$x_i$};
\node [above left] at (o2) {$x_j$};

\node [above] at (u) {$u$};
\node [below] at (v) {$v$};
\node [left] at (u1) {$w$};

\lattice{3}
\bdot{v}
\bdot{u}
\bdot{u1}
\end{tikzpicture}

%% file: intermediate2.tex
\begin{tikzpicture}
\coordinate (o) at (0,0);
\coordinate (o1) at (4,0);
\coordinate (o2) at (0,4);

\coordinate (fake) at (-1.75,4);
\node (f) [above left] at (fake) {$\invisible{f}$};

\coordinate (v) at (2,0);
\coordinate (u) at (1,3);
\coordinate (u1) at (0,2);

\draw[] (o) -- (3,0) -- (3,3) -- (0,3) -- cycle;

\begin{scope}
\clip (0, 3) rectangle (3, 0);
\levels{1}{1}{3}{red}
\end{scope}

\Square{o}
\node [below left] at (-0.1,-0.1) {\red{$F$}};

\draw[->,thick] (o) -- (o1);
\draw[->,thick] (o) -- (o2);
\node [below right] at (o1) {$x_i$};
\node [above left] at (o2) {$x_j$};

\node [above] at (u) {$u$};
\node [below] at (v) {$v$};
\node [left] at (u1) {$w$};

\lattice{3}
\draw[>=latex,->,ultra thick, blue] (u)--(0,2);
\draw[>=latex,->,ultra thick, blue] (0,2) -- (0, 1);
\draw[>=latex,->,ultra thick, blue] (0, 1)--(0,0);

\draw[>=latex,->,ultra thick, blue] (2, 0)--(1,0);
\draw[>=latex,->,ultra thick, blue] (1, 0)--(0,0);
\bdot{v}
\bdot{u}
\bdot{u1}
\end{tikzpicture}

%% file: verticeshc2.tex
\begin{tikzpicture}
\coordinate (a) at (-3,0);
\coordinate (b) at (-4.75,1);
\coordinate (c) at (-4.75,4);
\coordinate (d) at (-3,3);
\coordinate (e) at (-0,0);
\coordinate (f) at (-1.75,1);
\coordinate (g) at (-1.75,4);
\coordinate (h) at (-0,3);

\coordinate (x1) at (1,0);

\draw[->,thick] (a)--(x1);
\node [below right] at (x1) {$x_1$};

\draw[] (e) -- (a) -- (d) -- (h) -- cycle;
\draw[] (h) -- (d) -- (c) -- (g) -- cycle;

\draw[fill=red!15] (a) -- (b) -- (c) -- (d) -- cycle;
%\node [] at (-8,5) {\Large{$\{x : x_1 = 0\}$}};
\node [] at (-4,2.5) {$\red{F}$};

\draw[thin,black!30] (e)--(f);
\draw[thin,black!30] (f)--(b);
\draw[thin,black!30] (g)--(f);
\draw[red,dashed,thick] (-2,0) -- (-3.75,1) -- (-3.75,4) -- (-2,3) -- cycle;
\draw[red,dashed,thick] (-1,0) -- (-2.75,1) -- (-2.75,4) -- (-1,3) -- cycle;
\draw[red,dashed,thick] (e) -- (f) -- (g) -- (h) -- cycle;

\node (u) [below right] at (0,0) {$u$};
\node (v) [above left] at (c) {$v$};

%\node (w) [left] at (-7.17,2.3) {\huge{$u'$}};
%\bdot{-7.17,2.3}
\node (z) [left] at (-3,-0.08) {$u'$};
\bdot{-3,0}
\bdot{e}
\bdot{c}

%\draw[blue,ultra thick] (0,0)--(-2,0);
%\bdot{-2,0}
%\draw[blue,ultra thick] (-2,0)-- (-4,2);
%\bdot{-4,2}
%\draw[blue,ultra thick] (-4,2)-- (-7.17,2.3);

\draw[>=latex,->,blue,ultra thick] (0,0)--(-3,0);
%\draw[>=latex,->,blue,ultra thick] (-1,0)-- (-2,0);
%\draw[>=latex,->,blue,ultra thick] (-2,0)-- (-3,0);
\bdot{0,0}
\bdot{-1,0}
\bdot{-2,0}

\end{tikzpicture}

%% file: verticeshc1.tex
\begin{tikzpicture}
\coordinate (a) at (-3,0);
\coordinate (b) at (-4.75,1);
\coordinate (c) at (-4.75,4);
\coordinate (d) at (-3,3);
\coordinate (e) at (-0,0);
\coordinate (f) at (-1.75,1);
\coordinate (g) at (-1.75,4);
\coordinate (h) at (-0,3);

\coordinate (x1) at (1,0);

\draw[->,thick] (a)--(x1);
\node [below right] at (x1) {$x_1$};

\draw[] (e) -- (a) -- (d) -- (h) -- cycle;
\draw[] (h) -- (d) -- (c) -- (g) -- cycle;

\draw[fill=red!15] (a) -- (b) -- (c) -- (d) -- cycle;
%\node [] at (-8,5) {\Large{$\{x : x_1 = 0\}$}};
\node [] at (-4,2.5) {$\red{F}$};
\draw[thin,black!30] (e)--(f);
\draw[thin,black!30] (f)--(b);
\draw[thin,black!30] (g)--(f);
\draw[red,dashed,thick] (-2,0) -- (-3.75,1) -- (-3.75,4) -- (-2,3) -- cycle;
\draw[red,dashed,thick] (-1,0) -- (-2.75,1) -- (-2.75,4) -- (-1,3) -- cycle;
\draw[red,dashed,thick] (e) -- (f) -- (g) -- (h) -- cycle;

\node (u) [below right] at (0,0) {$u$};
\node (v) [above left] at (c) {$v$};

\node (w) [above] at (-3.5,1.15) {$u'$};
\bdot{-3.585,1.15}
%\node (z) [below left] at (a) {\huge{$u'$}};
%\bdot{-6,0}
\bdot{e}
\bdot{c}

\draw[>=latex,->,blue,ultra thick] (0,0)--(-1,0);
\bdot{-1,0}
\draw[>=latex,->,blue,ultra thick] (-1,0)-- (-2,1);
\bdot{-2,1}
\draw[>=latex,->,blue,ultra thick] (-2,1)-- (-3.585,1.15);

%\draw[green,ultra thick] (0,0)--(-2,0);
%\draw[green,ultra thick] (-2,0)-- (-4,0);
%\draw[green,ultra thick] (-4,0)-- (-6,0);
%\bdot{0,0}
%\bdot{-2,0}
%\bdot{-4,0}

\end{tikzpicture}